\documentclass[11pt]{article} 
\usepackage{hyperref} 
\usepackage{graphicx}
\usepackage{array}
\usepackage{amsthm}
\usepackage{amsmath}
\usepackage{amssymb}
\usepackage{microtype}
\usepackage{float}
\usepackage[hypcap]{caption}
\newtheorem{theorem}{Theorem}
\newtheorem{observation}{Observation}
\newtheorem{con}{Conjecture}
\newtheorem{lemma}{Lemma}
\begin{document}
\title{\textsc{Square-Root Finding Problem In Graphs, A~Complete Dichotomy Theorem.}}

\author{Babak Farzad$^{1}$ and Majid Karimi$^{2}$\\ 
\\\vspace{6pt} Department of Mathematics\\
Brock University, St. Catharines, ON., Canada\\
bfarzad@brocku.ca$^{1}$
mkarimi@brocku.ca$^{2}$}

\date{September, 2012}

\maketitle
\begin{abstract}
\noindent
Graph $G$ is the square of graph $H$ if two vertices $x,y$ have an edge in
  $G$ if and only if $x,y$ are of distance at most two in $H$.
Given $H$ it is easy to compute its square $H^2$. Determining	
  if a given graph $G$ is the square of some graph is not easy in general.
  Motwani and Sudan \cite{MotSud1994} proved that it is NP-complete
  to determine if a given graph $G$ is the square of some
  graph.
  The graph introduced in their reduction is a graph that contains
  many triangles and is relatively dense.
  Farzad et al. \cite{Far2009} proved the NP-completeness for finding a square root for
  girth $4$ while they
  gave a polynomial time algorithm for computing a
  square root of girth at least six.
  Adamaszek and Adamaszek \cite{Adam2011}
  proved that if a graph has a square root of girth six
  then this square root is unique up to isomorphism.
  In this paper we consider the characterization and recognition
  problem of graphs that are square of graphs of girth at least five.
  We introduce a family of graphs with exponentially many non-isomorphic square roots,
  and as the main result of this paper we prove that the square root finding
  problem is NP-complete for square roots of girth five. This proof
  is providing the complete dichotomy theorem for square root
  problem in terms of the girth of the square roots.
  \end{abstract}
\section{Introduction}
Graph $G$ is called the $r^{th}$ \textit{power} of $H$ and $H$ is
called an $r^{th}$ \textit{root} of $G$, if $v$ is adjacent to $u$ in $G$
if and only if $d_{H}(v,u)\leq r$, where $d_{H}(v,u)$ is the distance between
$u$ and $v$ in graph $H$.
We are interested in the characterization and recognition of square graphs.
\textit{Root} and \textit{root finding} are concepts familiar to most
branches of mathematics.
Root and root finding for graph also is a basic operation in graph theory.
The complexity problem of root finding for graphs is
an extensively studied problem in algorithmic graph theory.

\par The main motivation for studying the complexity of checking if a given graph is
a certain power (square specifically) of another graph comes from
distributed computing. In a model introduced by Linial \cite{Linial1992},
the $r^{th}$ power of graph $H$ may represents the possible
flow of information in $r$ round of communication in a distributed network of
processors organized according to $H$. He introduced
a question about the characterization of this problem which is solved by
Motvani and Sudan \cite{MotSud1994}.

Mukhopadhyay \cite{Muk1967} showed that a graph $G$ is the
square of some graph if and only if
there exists a complete induced subgraph $G_i$ corresponding to each vertex
$v_i$ such that
\begin{itemize}
\item[-] $v_i \in G_i$;
\item[-] $v_i \in G_j$ if and only if 	$v_j \in G_i$;
\item[-] $\bigcup G_i=G$.
\end{itemize}
However Mukhopadhyay's theorem contains different aspects of the
\textit{maximum clique} problem which is an NP-hard problem.
Hence it does not benefit the study from a complexity point of view.

Ross and Harary \cite{RosHar1960} characterized squares of trees and showed
that tree square roots, when they exist, are unique up to isomorphism.
Motwani and Sudan \cite{MotSud1994} proved that it is NP-complete
to determine if a given graph has a square.
The graph introduced in their reduction is a graph that contains
many triangles and is relatively dense.
On the other hand, there are polynomial time algorithms
to compute the tree square root
\cite{LinSki1995,KeaCor1998,Lau2006,BraLeSri2006,ChaKoLu2006},
a bipartite square root \cite{Lau2006},
and a proper interval square root \cite{LauCor2004}.
Farzad et al. \cite{Far2009} provided an \textit{almost} dichotomy theorem
for the complexity of the recognition
problem in terms of \textit{the girth of the square roots}.
They provided a polynomial time characterization of square of
graphs with girth at least $6$. They proved that the square root (if it exists)
is unique up to isomorphism when the girth of square root is at least $7$.
They also proved the NP-completeness of the problem for square roots of
girth $4$. Adamaszek and Adamaszek~\cite{Adam2011} proved that 
the square root of a graph is unique up to isomorphism when the girth of
square root is at least $6$ if it exits.

A summary of the the study for square root finding problem
in terms of the girth of the square root
is presented in Table~\ref{table:SSF} (in this table $H$ is indicating a square root graph).

\begin{center}
\begin{table}[h]
\begin{center}
\begin{tabular}{|c|c|c|}
\hline
Girth&Complexity Class & Unique up to isomorphism\\
\hline
$g(H) = \infty$&$O(|V|+|E|)$ \cite{ChaKoLu2006}&Yes\\
\hline
$g(H)\geq 7$&$O(|V|\times|E|)$ \cite{Far2009}&Yes \cite{Far2009}\\
\hline
$g(H)\geq 6$&$O(|V|\times|E|)$ \cite{Far2009}&Yes \cite{Adam2011}\\
\hline
$g(H) \geq 5$&?&No\\
\hline
$g(H) \geq 4$&NP-complete \cite{Far2009}&No\\
\hline
$g(H) \geq 3$&NP-complete \cite{MotSud1994}&No\\
\hline
\end{tabular}
\vspace{.4cm}
\end{center}
\caption{Complexity of square root problem in terms of girth of square roots.}
\label{table:SSF}
\end{table}
\end{center}

The recognition problem has been open for square roots of girth $5$.
In Section~\ref{sec:girthfive} we show that this problem is NP-complete.
The result is providing a complete dichotomy complexity theorem for
square root problem.
We also generalize the graph introduced in \cite{Adam2011} to
construct a family of graphs with exponential number of
non-isomorphic square roots.


\par \textbf{Definitions and notations:}
All graphs considered are finite, undirected and simple.
Let $G=(V_G, E_G)$ be a graph with vertex set in $V_G$ and edge set $E_G$.
We denote the adjacency of two vertices $u$ and $v$ in graph
$G$, by $u \overset{G}{\sim} v$.
To show that $v$ is adjacent to every element of a set $A \subseteq V(G)$,
we use $v \overset{G}{\sim} A$.
The \emph{neighbourhood} of a vertex $v$ in graph $G$ denoted by $N_G(v)$
is the set all vertices in $G$ adjacent to $v$.
The \textit{closed neighbourhood} of $v$ in $G$ denoted by $N_G[v]$,
is its neighbourhood containing $v$ as well, i.e. $N_G[v]=N_G(v)\cup \{v\}$.
The cardinality of the set $N_{G}(v)$ is called
the \emph{degree} of $v$ in $G$. The minimum degree
of a graph $G$ is shown by $\delta_G$.

Let $d_G(x,y)$ be the length i.e., number of edges
of a shortest path in $G$ between $x$ and $y$.
Let $G^k=(V_G,E^k)$ with $xy \in E^k$ if and only if
$1\le d_G(x,y) \le k$, denote the {\em $k$-th power of $G$}.
If $G=H^k$ then $G$ is the $k$-th power of the graph $H$ and
$H$ is a \emph{$k$-th root} of $G$. Since the power of a graph $H$
is the union of the powers of the connected components of $H$,
we may assume that all graphs considered are connected.

A set of vertices $Q\subseteq V_G$ is called a \emph{clique} in $G$ if every two
distinct vertices in $Q$ are adjacent; a \emph{maximal clique} is a clique that is not
properly contained in another clique.
Given a set of vertices $X\subseteq V_G$, the subgraph induced by
$X$ is denoted by $G[X]$ and $G-X$ stands for $G[V(G) \setminus X]$.
If $X=\{a, b, c, \ldots\}$, we write $G[a, b, c, \ldots]$ for $G[X]$.
Also, we often identify a subset of vertices with the
subgraph induced by that subset, and vice versa.

The \textit{girth} of $G$, $g(G)$, is the smallest length of a cycle in $G$; in
case $G$ has no cycles, we set $g(G)=\infty$. In other words, $G$ has girth
$k$ if and only if $G$ contains a cycle of length $k$ but
does not contain any cycle of length $\ell=3, \ldots, k-1$.

A complete graph is one in which every two distinct vertices are
adjacent; a complete graph on $n$ vertices is also denoted by $K_n$.
A \emph{star} is a graph with \emph{at least two} vertices that
has a vertex adjacent to all vertices and the other vertices are pairwise non-adjacent.
A star on $n+1$ vertices is denoted by $K_{1,n}$.

Two graphs $G_1$ and $G_2$ are called \textit{isomorphic} when there is a
bijection $f$ from $V_{G_1}$ to $V_{G_2}$ such that for all vertices
$u,v \in V_{G_1}$: $u \overset{G_1}{\sim} v$ if and only if $f(u) \overset{G_2}{\sim} f(v)$.

\section{Graphs with Many Non-Isomorphic Square Roots of Girth Five}
\label{sec:many}
For a given graph $G$ if there exists $H$ where $G=H^2$ and
$g(H)\geq 6$, then $H$ is unique up to isomorphism \cite{Adam2011}.
However this is not true when the girth of $H$ is at least $5$.
For $G=K_5$, two graphs $K_{1,4}$ and $C_{5}$ are 
non-isomorphic square roots of $G$. These two graphs can be
used to introduce a family of non-isomorphic pairs of graphs with
the same square, see Figure~\ref{hhh}.

\begin{figure}[h]
\begin{center}
\includegraphics[scale=.43]{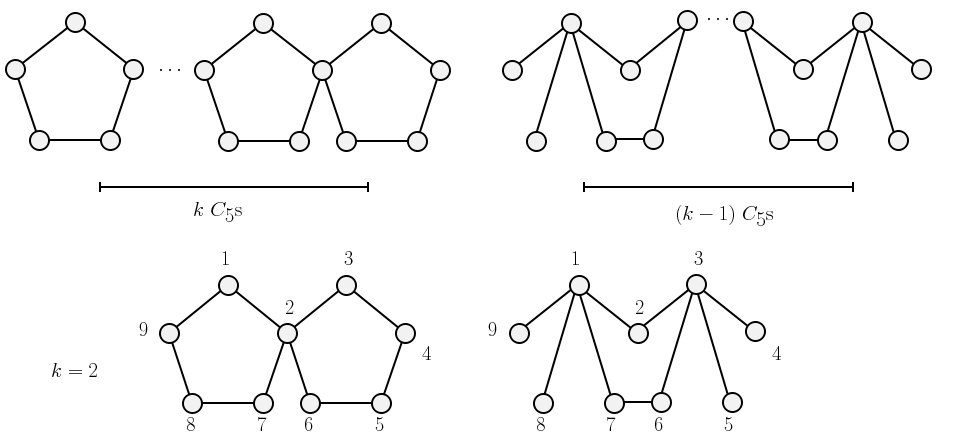}
\end{center}
\caption{A family of non-isomorphic graphs with identical square graph.}
\label{hhh}
\end{figure}

\par Notice that graphs in this family contain
vertices of degree $1$. Such vertices were a main
source of technicalities in the past studies.
\par In \cite{Adam2011} there is also an example of
a graph with two non-isomorphic square root of girth five, see Figure~\ref{fig:gonetwo}.
These two graphs are more interesting as, unlike graphs shown in
Figure~\ref{hhh}, they contain no vertex of degree $1$.
These two graphs are also the smallest non-isomorphic graphs with girth five,
minimum degree $2$ and identical squares.
In this paper, we call these two graphs $\mathcal{G}_1$ and $\mathcal{G}_2$.

\begin{figure}[H]
\begin{center}
\includegraphics[scale=.4]{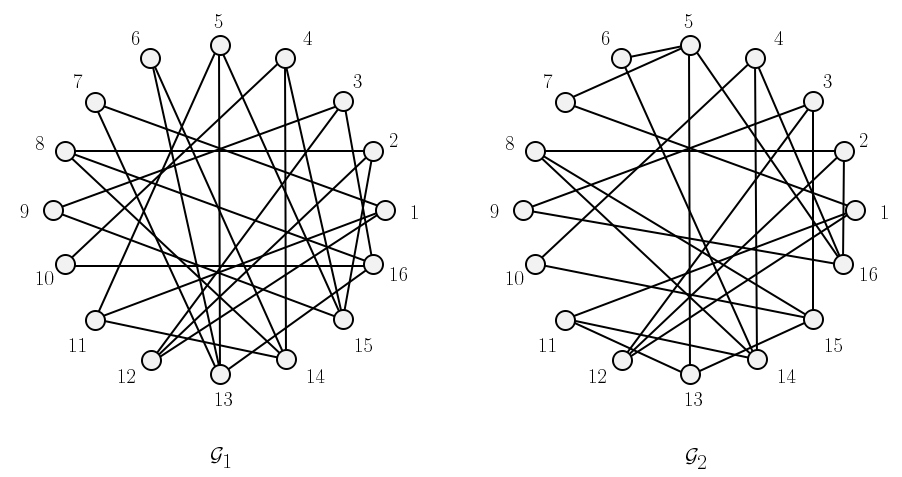}
\end{center}
\caption{Non-isomorphic graphs $\mathcal{G}_1$ and $\mathcal{G}_2$ with no vertex of degree $1$ and identical squares.}
\label{fig:gonetwo}
\end{figure}

\par It is also an
interesting question (from a complexity point of view)
to ask if there exists a 
graph with many non-isomorphic square roots.
We show that $\mathcal{G}_1$ and $\mathcal{G}_2$ can be used to construct
a family of graphs with many non-isomorphic square roots.
With current labelling of $\mathcal{G}_1$ and $\mathcal{G}_2$,
we have three vertices $1,12$
and $14$, that their neighbourhoods in both $\mathcal{G}_1$ and $\mathcal{G}_2$ are identical.
So we may identify two graphs on one of these three vertices to construct
a new graph with more than one square roots.
For example, we can identify vertex $1$ in both
$\mathcal{G}_1$ and $\mathcal{G}_2$
as shown in Figure~\ref{fig:geminifam}.

\begin{figure}[h]
\begin{center}
\includegraphics[scale=.38]{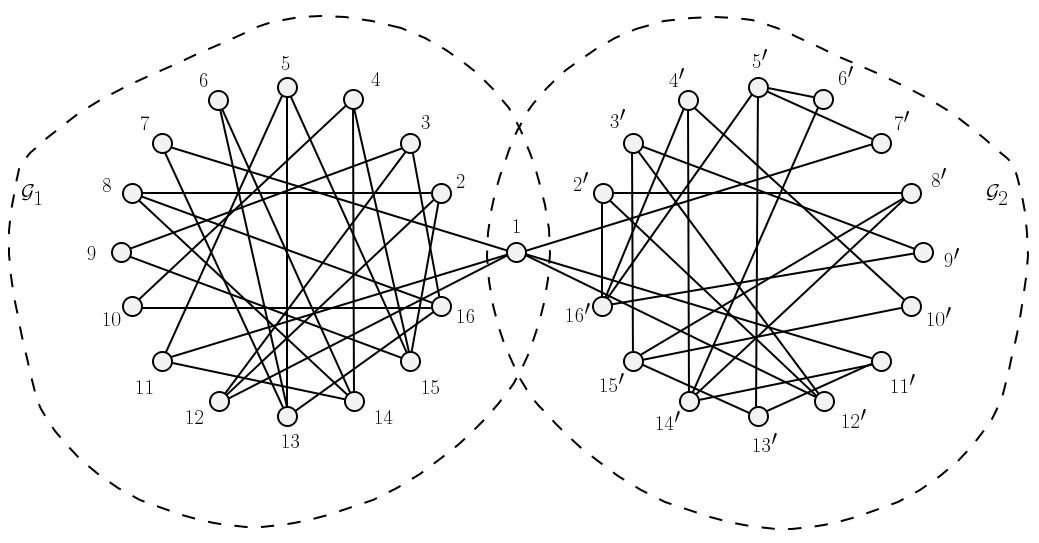}
\end{center}
\caption{Connecting $\mathcal{G}_1$ and $\mathcal{G}_2$
by identifying vertex $1$.}
\label{fig:geminifam}
\end{figure}

\begin{observation}
The square of the graph shown in Figure~\ref{fig:geminifam}
has three non-isomorphic square roots.
\end{observation}
\begin{proof}
In Figure~\ref{fig:geminifam}, by replacing the copy of $\mathcal{G}_1$
on the vertices $\{1,2,\ldots,16\}$ with a copy of $\mathcal{G}_2$,
we would get a different graph with the same square.
Hence switching copies of $\mathcal{G}_1$ and $\mathcal{G}_2$
constructs three non-isomorphic graphs
with identical squares.
\end{proof}

The process of connecting $\mathcal{G}_1$s and $\mathcal{G}_2$s
by identifying one of those three vertices
can form a family of graphs with girth five, minimum degree of $2$
and exponentially many non-isomorphic square roots.
See Figure~\ref{fig:geminifam1} for an illustration of $16$ non-isomorphic graphs
with identical square.
\begin{figure}[h]
\begin{center}
\includegraphics[scale=.4]{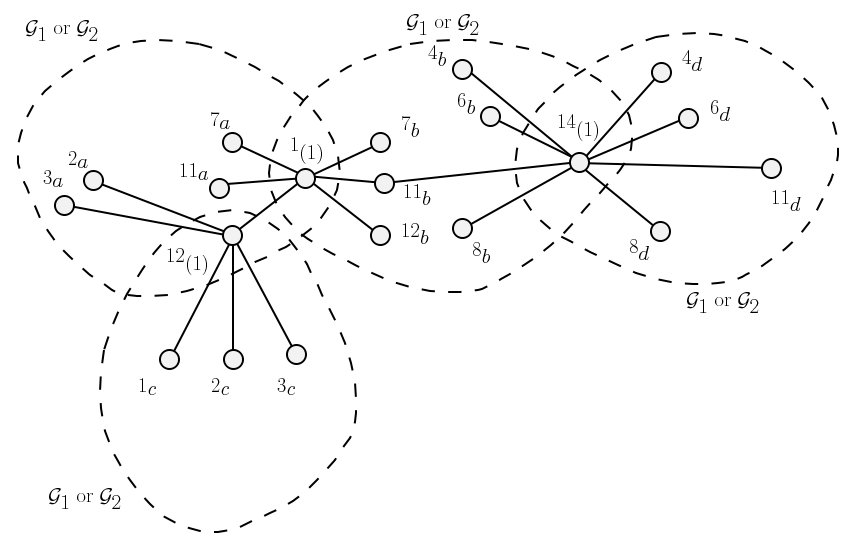}
\end{center}
\caption{Non-isomorphic graphs with identical square.}
\label{fig:geminifam1}
\end{figure}

\par This process is introducing a family of graphs 
with exponentially many non-isomorphic square roots.
This family indicates that even with the restriction $\delta_{H} \geq 2$ knowledge of
any local neighbourhood is not sufficient to reconstruct the rest of the square root graph.
\par When a square root graph has no short cycle (girth of at least $6$)
square root finding problem is solvable
by an efficient algorithm \cite{Far2009}.
The main idea of this algorithm (and almost
all attempts to find an efficient algorithm for square root finding problem)
is to use a known neighbourhood
of the square root graph and reconstruct the
whole square root graph by only using
informations from the square graph. Indeed if we know
an arbitrary neighbourhood of graph $H$ of girth at least six,
where $H^2=G$, then we can
recognize second neighbours (vertices of distance two) of that vertex.
In this way the whole graph $H$ can be uniquely reconstructed
with only using information of $G$.
The family of graphs we introduced using $\mathcal{G}_1$ and $\mathcal{G}_2$
indicates that by knowing an arbitrary neighbourhood of the square root graph
we can never decide the rest of the graph, as there are always options 
(to decide a second neighbourhood of a vertex)
that results different (non-isomorphic) graphs. Hence knowing a constant number of
neighbourhoods in the square root graph can not help to find a square root
for a given graph (or to decide if there exists a square root graph).
\par We also use $\mathcal{G}_1$ and $\mathcal{G}_2$
graphs as part of our reduction in Section~\ref{sec:girthfive}.
We need to show that the graph $\mathfrak{G}={\mathcal{G}_{1}}^2={\mathcal{G}_{2}}^2$
has only two non-isomorphic square roots which are $\mathcal{G}_1$ and $\mathcal{G}_2$.
For the rest of this paper we use $\mathfrak{G}$ as the square of ${\mathcal{G}_{1}}$ (or ${\mathcal{G}_{2}})$.
\begin{theorem}
Let $\mathfrak{G}=H^2$ for $g(H)=5$, then $H$ is either isomorphic to $\mathcal{G}_1$ or to $\mathcal{G}_2$.
\label{lemma:uniquepairofgeminis}
\end{theorem}
A proof of this theorem can be found in Appendix-A.
\section{Square of graphs with girth five}
\label{sec:girthfive}
In this section we show that the following problem is NP-complete.\\

\begin{center}
\textsc{Square of Graphs With Girth Five}\\
\begin{tabular}{rp{10cm}}
\textit{Instance}& A graph $G$.\\
\textit{Question:}& Does there exists a graph $H$ with girth at least $5$ such that $G=H^2$?
\end{tabular}
\end{center}
It is an easy observation that \textsc{Square of Graphs With Girth Five}
is in $NP$. We will reduce a variation of
the ``\textit{positive 1-in-3 SAT}'' problem (which is an NP-complete problem \cite{SCHAEFER1978}) to
\textsc{Square of Graphs With Girth Five}.
Positive 1-in-3 SAT is a variant of the 3-satisfiability problem (3SAT).
Like 3SAT, the input instance is a collection of clauses, where each
clause is the disjunction of exactly three literals, and each literal is just a variable
(there are no negations, which is why it is called positive).
The positive 1-in-3 3SAT problem is to determine whether
there exists a truth assignment to the variables so that each clause
has exactly one true variable (and thus exactly two false variables).
In this paper we are interested in another variation
of the positive 1-in-3 SAT,  which we call it
\textit{POSITIVE AND
MINIMUM INTERSECTING 1-in-3 SAT}.

\begin{center}
\textsc{Positive and Minimum Interesting 1-in-3 SAT}.
\begin{tabular}{rp{10cm}}
\textit{Instance:}& A collection of clauses, where each
clause is the disjunction of exactly three variables
and two different clauses are sharing at most
one variable.\\
\textit{Question:}& Does there exists a truth assignment to
the variables so that each clause has exactly one true variable?
\end{tabular}
\end{center}

\begin{theorem}
\textsc{Positive and Minimum Interesting 1-in-3 SAT} is NP-complete.
\end{theorem}
\begin{proof}
It is trivial that this problem is in NP.
We reduce an instance of a \textit{Positive 1-in-3 SAT} to a
\textsc{Positive and Minimum Interesting 1-in-3 SAT}.
Let $\phi$ be a given collection of clauses as an instance of the
positive 1-in-3 SAT.

For each pair of clauses $c: (x \vee y \vee  z)$ and
$d:(x\vee  y\vee  u)$ in $\phi$, that are sharing two
variables $x$ and $y$, we know $u$ and $v$ must have the same truth value.
So we may identify the two variables and thus replace $v$
with $u$ and remove the clause $d$.
We construct $\phi ^\prime$ from $\phi$ by removing one
of clauses in each pair of clauses that are sharing two variable.
Therefore $\phi ^\prime$ is an instance of 
\textsc{Positive and Minimum Interesting 1-in-3 SAT}.
This reduction shows that
\textsc{Positive and Minimum Interesting 1-in-3 SAT}
is NP-complete.
\end{proof}

In this section we reduce the \textsc{Positive and Minimum Interesting 1-in-3 SAT} to
\textsc{Square of Graphs With Girth Five}.
\subsection{The Reduction}
Before introducing the reduction in all details, we present three main ideas
of the graph construction that we will explain below. For convenience,
we represent $\forall a\in A: v \sim a$ by $v \sim A$, and also $\{x_a,y_a,z_a,\ldots\}$
by $\{x,y,z,\ldots\}_{a}$.

\par First is the idea of using graph $\mathfrak{G}$ to represent each copy of a variable.
As we proved in Appendix A, a square root of $\mathfrak{G}$ is a graph which is isomorphic to
either $\mathcal{G}_1$ or $\mathcal{G}_2$.

We set $\mathcal{G}_1$ to represent the FALSE value and 
$\mathcal{G}_2$ to represent the TRUE value. 
If the square root of the subgraph that is representing a copy of a variable $x$ 
is isomorphic to $\mathcal{G}_1$ we conclude that $x$ is $FALSE$.
Otherwise, that is if it is isomorphic to $\mathcal{G}_2$, we conclude that $x$ is $TRUE$.
\par The second idea is to represent a clause $c_i = x_i \vee y_i \vee z_i$
in such a way that exactly one of $x_i, y_i$ and $z_i$
is true (i.e., exactly one of the subgraphs that are representing 
the three variables is isomorphic to $\mathcal{G}_2$ and the other two are isomorphic
to $\mathcal{G}_1$). For this, for each clause $c_i$ we introduce four new vertices
$y_{1}^{i},\ldots,y_{4}^{i}$ to construct a Petersen graph in the square root
(that is a $K_{10}$ in the square graph) using vertices
$5$ and $13$ in the three subgraphs representing the copies of variables in $c_i$.
This construction is illustrated in Figure~\ref{fig:clause}.
\begin{figure}
\begin{center}
\includegraphics[scale=.42]{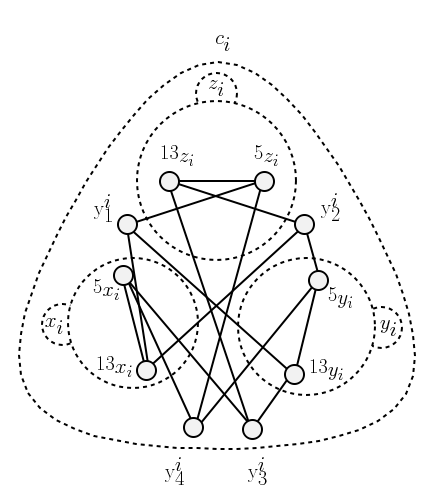}
\caption{Structure of a subgraph of the square root graph which represents a clause.}
\label{fig:clause}
\end{center}
\end{figure}

\begin{lemma}
\label{lem:clause}
The square of the graph shown in Figure~\ref{fig:clause}
has three different (up to labelling) square roots.
The other two square roots can be obtained by switching
$\mathcal{G}_1$s with $\mathcal{G}_2$s.
However, it has a unique square root of
girth $5$ up to isomorphism.
\end{lemma}
\begin{proof}
Let $\mathfrak{C}$ be the square of
the graph shown in Figure~\ref{fig:clause} on $X \cup Y \cup Z \cup I $ where
$X=\{1,2,\ldots,16\}_{x_i}, Y=\{1,2,\ldots,16\}_{y_i}, Z=\{1,2,\ldots,16\}_{z_i}$
and $I=\{y_{1}^{i},\ldots, y_{4}^{i}\}$.
Also let $D$ to be a square root of $\mathfrak{C}$.
Graphs constructed by switching
$\mathcal{G}_1$s with $\mathcal{G}_2$s.
The isomorphism of these three graphs can be obtained by a permutation on $I$.

For example, assume that $D[X] \cong D[Y] \cong \mathcal{G}_1$
and $D[Z] \cong \mathcal{G}_2$. Then the graph obtained by the permutation
$y_{1}^{i} \leftrightarrow y_{3}^{i}$ and $y_{2}^{i} \leftrightarrow y_{4}^{i}$
has the same square as the graph shown in Figure~\ref{fig:clause}.\\
By Theorem~\ref{lemma:uniquepairofgeminis},
the square root of the subgraph induced by $X, Y$ or $Z$
is either $\mathcal{G}_1$ or $\mathcal{G}_2$.
Now consider the neighbourhoods of vertices $5$ and $13$.
We have $N_{\mathcal{G}_1}[5]=N_{\mathcal{G}_2}[13]=\{5,13,11,15\}$
and $N_{\mathcal{G}_1}[13]=N_{\mathcal{G}_2}[5]=\{5,13,6,7,16\}$. 
It can be seen that if none or more than one of
the square roots of the subgraph induced by $X, Y$ or $Z$ is isomorphic to $\mathcal{G}_2$,
then there would be no permutation on $I$, that form
the same square as the graph shown in Figure~\ref{fig:clause}.
\end{proof}

\par The third idea is to make sure that different copies
of the same variable have the same truth value.
Again we use the fact that 
$N_{\mathcal{G}_1}[5]=N_{\mathcal{G}_2}[13]=\{5,13,11,15\}$
and $N_{\mathcal{G}_1}[13]=N_{\mathcal{G}_2}[5]=\{5,13,6,7,16\}$. 
Let $x_i$ and $x_j$ be two copies of the same variable in two different clauses $c_i$
and $c_j$.
We introduce two new vertices
called $v_{x_i, x_j}$ and $w_{x_i, x_j}$ which form a 
$C_6$ in the square root graph together with
the vertices $5$ and $13$ in the subgraphs corresponding to $x_i$ and $x_j$.
If both $x_i$ and $x_j$ are TRUE then $v_{x_i, x_j} \sim \{13_{x_i},13_{x_j}\}$
and $w_{x_i, x_j} \sim \{5_{x_i},5_{x_j}\}$,
otherwise $w_{x_i, x_j} \sim \{13_{x_i},13_{x_j}\}$
and $v_{x_i, x_j} \sim \{5_{x_i},5_{x_j}\}$.
This construction is shown in Figure~\ref{fig:copiesofthesamevariable}.
Moreover we have the following Lemma.

\begin{figure}
\begin{center}
\includegraphics[scale=.38]{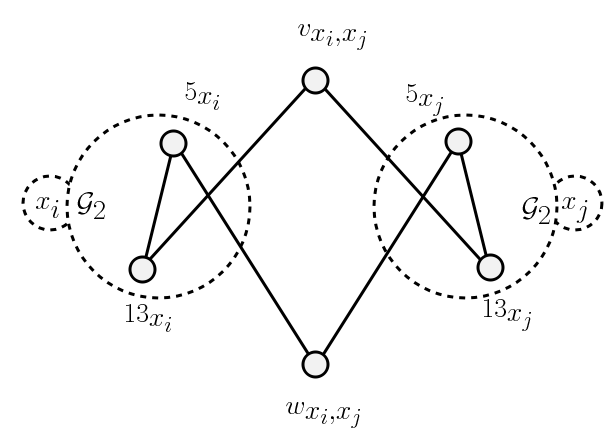}
\caption{Different copies of a variable have the same truth value.}
\label{fig:copiesofthesamevariable}
\end{center}
\end{figure}

\begin{lemma}
\label{lem:copiesofavariable}
Let $\mathfrak{X}$ be the square of the graph shown in Figure~\ref{fig:copiesofthesamevariable}
on the vertex set of $X^i\cup X^j \cup \{v_{x_i, x_j},w_{x_i, x_j}\}$
where $X^i=\{1,2,\ldots,16\}_{x_i}$ and $X^j=\{1,2,\ldots,16\}_{x_j}$.
Let $\mathsf{X}^2=\mathfrak{X}.$
If $\mathsf{X}[X^i] \cong\mathcal{G}_1$ (or $\mathcal{G}_2$) then 
$\mathsf{X}[X^j] \cong \mathcal{G}_1$ (or $\mathcal{G}_2$).
\end{lemma}
\begin{proof}
Assume otherwise and let (without loss of generality)
$\mathsf{X}[X^i] \cong\mathcal{G}_1$ while
$\mathsf{X}[X^j] \cong \mathcal{G}_2$, hence $v_{x_i, x_j}$
must be adjacent to $5_{x_i}$ and $13_{x_j}$ which means
$5_{x_i} \overset{\mathfrak{X}}{\sim} 13_{x_j}$, and this is a contradiction
as $5_{x_i} \overset{\mathfrak{X}}{\nsim} 13_{x_j}$.
\end{proof}

\par{\bf \textit{Reduction Graph:}} Let $\phi: (c_1 \wedge c_2 \wedge \cdots \wedge c_n)$ be an instance of
\textsc{Positive and Minimum Intersecting 1-in-3 SAT}
such that $c_i= x_{i} \vee y_i \vee z_i$.
As a convention we use $x_i$ and $x_j$ to represents two copies of variable $x$
in distinct clauses $c_i$ and $c_j$.

\par We construct an instance $G = G(\phi)$ and we show that there exists a square
root $H$ of of girth $5$ of graph $G$ corresponds to a satisfying assignment of $\phi$.
\par The vertex set of graph $G(\phi)$ consists of:
	\begin{itemize}
	\item For every copy $x_i$ of variable $x$, $V_{x_i}=V_{\mathcal{G}_1}(=V_{\mathcal{G}_2})=\{1,2,\ldots, 16\}_{x_i}$,
	representing $16$ vertices of a graph $\mathfrak{G}$.
	
	\item For each clause $c_i$, $V_i=\{y_{1}^{i},y_{2}^{i},y_{3}^{i},y_{4}^{i}\}$.
	
	\item $W_{x_i,x_j}=\{v_{x_i,x_j},w_{x_i,x_j}\}$, corresponding
	to two copies $x_i$ and $x_j$ of the same variable $x$,
	in two distinct clauses $c_i$ and $c_j$.
	\end{itemize}
	The edge set of $G(\phi)$ consists of:
	\begin{itemize}
	\item Variable edges: for each $x_i$, $G[V_{x_i}]=\mathfrak{G}$.
	\item Clause edges: For each clause $c_i=x_i \vee y_i \vee z_i$\\
	$G[\{5_{x_i},13_{x_i},5_{y_i},13_{y_i},5_{z_i},13_{z_i},y_{1}^{i},
	y_{2}^{i},y_{3}^{i},y_{4}^{i}\}] \cong K_{10}$, i.e., they are all adjacent to each other. Also by recalling that $N_{\mathcal{G}_1}[5]=N_{\mathcal{G}_2}[13]=\{5,13,11,15\}$
and $N_{\mathcal{G}_1}[13]=N_{\mathcal{G}_2}[5]=\{5,13,6,7,16\}$,
	we have:\\
	$y_{1}^{i} \sim \{11_{x_i},15_{x_i}, 6_{y_i},7_{y_i},16_{y_i}, 11_{z_i},15_{z_i}\}$,\\
	$y_{2}^{i} \sim \{6_{x_i},7_{x_i},16_{x_i}, 11_{y_i},15_{y_i}, 11_{z_i},15_{z_i}\}$,\\
	$y_{3}^{i} \sim \{6_{x_i},7_{x_i},16_{x_i}, 6_{y_i},7_{y_i},16_{y_i}, 6_{z_i},7_{z_i},16_{z_i}\}$,\\
	$y_{4}^{i} \sim \{11_{x_i},15_{x_i}, 6_{y_i},7_{y_i},16_{y_i}, 6_{z_i},7_{z_i},16_{z_i}\}$, see~Figure~\ref{CI}.\\
	\begin{figure}
	\begin{center}
	\includegraphics[scale=.35]{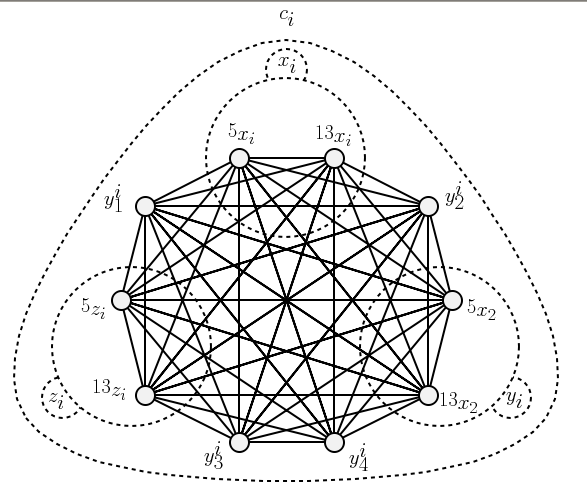}
	\caption{A subgraph of $G(\phi)$ corresponding to a clause.}
	\label{CI}
	\end{center}
	\end{figure}
	\item Intra clause edges: for each clause $c_i=x_i \vee y_i \vee z_i$ where $i \notin \{j,k,m\}$:\\
	$y_{1}^{i} \sim \{v_{x_i,x_j},w_{y_i,y_k}, v_{z_i,z_m}\}$,\\
	$y_{2}^{i} \sim \{w_{x_i,x_j},v_{y_i,y_k}, v_{z_i,z_m}\}$,\\
	$y_{3}^{i} \sim \{w_{x_i,x_j},w_{y_i,y_k}, w_{z_i,z_m}\}$,\\
	$y_{4}^{i} \sim \{v_{x_i,x_j},v_{y_i,y_k}, w_{z_i,z_m}\}$.\\
	Notice that we may have only a subset of these edges depending on the existence of
	$x_j$ (the copy of variable $x$ in $c_j$),
	$y_k$ (the copy of variable $y$ in $c_k$) and
	$z_m$ (the copy of variable $z$ in $c_m$).
	\item Edges for different copies of a variable:
	for each arbitrary pair $x_i$ and $x_j$ which are different copies of the
	same variable, \\
	$v_{x_i,x_j} \sim \{13_{x_i}, 13_{x_j},5_{x_i}, 5_{x_j}\}$,\\
	$v_{x_i,x_j} \sim \{11_{x_i}, 15_{x_i}, 11_{x_j}, 15_{x_j}\}$,\\
	$w_{x_i,x_j} \sim \{13_{x_i}, 13_{x_j},5_{x_i}, 5_{x_j}\}$,\\
	$w_{x_i,x_j} \sim \{6_{x_i}, 7_{x_i}, 16_{x_i}, 6_{x_j}, 7_{x_j}, 16_{x_j}\}$,~see Figure~\ref{fig:Two copies of a variable}.\\
	\begin{figure}
	\begin{center}
	\includegraphics[scale=.42]{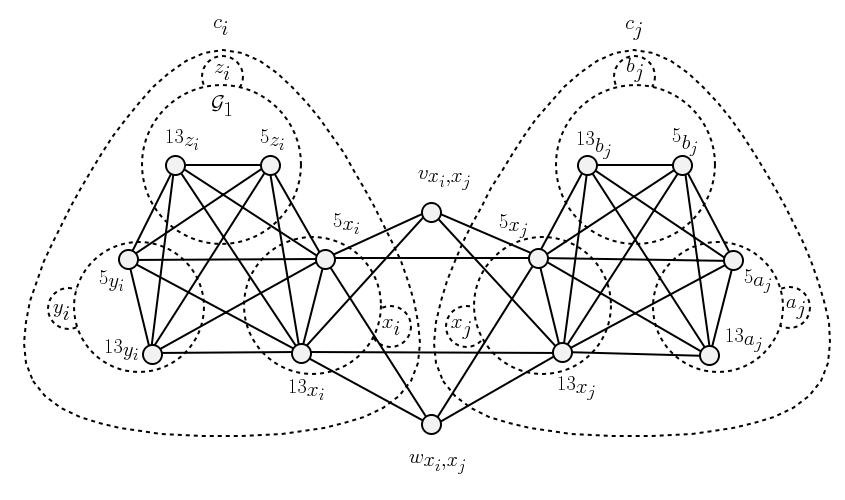}
	\caption{A subgraph of $G(\phi)$ corresponding to the clause $c_i=x_i \vee y_i \vee z_i$.}
	\label{fig:Two copies of a variable}
	\end{center}
	\end{figure}
	\item Edges of variable copies:\\
	for an arbitrary variable $x$ and all $i\neq j$ and $k \neq l$, we have $w_{x_i,x_j} \sim w_{x_k,x_l}$
	and $v_{x_i,x_j} \sim v_{x_k,x_l}$.
	\end{itemize}
	It is an easy observation to see that $G(\phi)$ can be be
	constructed from $\phi$ in polynomial time.

\begin{lemma}
\label{lemma:SATtoSquare}
There exists a truth assignment to variables in instance
$\phi$ of \textit{POSITIVE AND MINIMUM INTERSECTING 1-in-3 SAT}
that satisfies the formula if and only if
there exists a graph $H$
of girth five such that $G(\phi)=H^2$.
\end{lemma}
\begin{proof}
\begin{itemize}
\item \textit{Satisfiability to squareness:}
\begin{itemize}
\item[-] $H$ construction: we construct the graph $H$ by using a satisfying assignment of $\phi$ as follows:\\
\begin{itemize}
\item  For all $i$ such that
there exists a clause $c_i$ where $x_i \in c_i$,
$H[\{1,2,\ldots,16\}_{x_i}] =\mathcal{G}_2$ if $x$ is true
and $H[\{1,2,\ldots,16\}_{x_i}] =\mathcal{G}_1$
if $x$ is false.

\item For each pair of $x_i$ and $x_j$ where $i\neq j$
if $x$ is true then $v_{x_i,x_j} \overset{H}{\sim} \{13_{x_i},13_{x_j}\}$ and
$w_{x_i,x_j} \overset{H}{\sim} \{5_{x_i},5_{x_j}\}$.
Otherwise, that is if $x$ is false, $w_{x_i,x_j} \overset{H}{\sim} \{13_{x_i},13_{x_j}\}$ and
$v_{x_i,x_j} \overset{H}{\sim} \{5_{x_i},5_{x_j}\}$.

\item For each clause $c_i=x_i \vee y_i \vee z_i$:\\
if $x_i$ is true then\\
$y_{1}^{i} \overset{H}{\sim} \{13_{x_i},13_{y_i},5_{z_i}\}$, 
$y_{2}^{i} \overset{H}{\sim} \{5_{x_i},5_{y_i},5_{z_i}\}$, 
$y_{3}^{i} \overset{H}{\sim} \{5_{x_i},13_{y_i},13_{z_i}\}$, 
$y_{4}^{i} \overset{H}{\sim} \{13_{x_i},5_{y_i},13_{z_i}\}$.\\

if $y_i$ is true then \\
$y_{1}^{i} \overset{H}{\sim} \{5_{x_i},5_{y_i},5_{z_i}\}$, 
$y_{2}^{i} \overset{H}{\sim} \{13_{x_i},13_{y_i},5_{z_i}\}$, 
$y_{3}^{i} \overset{H}{\sim} \{13_{x_i},5_{y_i},13_{z_i}\}$, 
$y_{4}^{i} \overset{H}{\sim} \{5_{x_i},13_{y_i},13_{z_i}\}$.\\

if $z_i$ is true then\\
$y_{1}^{i} \overset{H}{\sim} \{5_{x_i},13_{y_i},13_{z_i}\}$, 
$y_{2}^{i} \overset{H}{\sim} \{13_{x_i},5_{y_i},13_{z_i}\}$, 
$y_{3}^{i} \overset{H}{\sim} \{13_{x_i},13_{y_i},5_{z_i}\}$, 
$y_{4}^{i} \overset{H}{\sim} \{5_{x_i},5_{y_i},5_{z_i}\}$.\\

Recall that in all cases $10$ vertices $y_{i}^{1},y_{i}^{2},y_{i}^{3}, y_{i}^{4}$,
$5_{x_i}, 5_{y_i}, 5_{z_i}, 13_{x_i}, 13_{y_i}, 13_{z_i}$
form a Petersen graph in $H$.
\end{itemize}
\item[-] $H^2=G(\phi)$: trivial.
\end{itemize}
\item \textit{Squareness to satisfiability:}\\
Let $H$ be a square root of $G(\phi)$.
By Theorem~\ref{lemma:uniquepairofgeminis}, graph $H[V_{x_i}]$
(for each copy of an arbitrary $x$) is isomorphic either to $\mathcal{G}_1$
or $\mathcal{G}_2$. We set $x$ to be true when $H[V_{x_i}] \cong \mathcal{G}_2$
and false otherwise. By Lemma~\ref{lem:copiesofavariable} all other copies of $x$
would also have the same truth value. By Lemma~\ref{lem:clause}
this assignment is a truth assignment to $\phi$ since exactly one variable in each clause is evaluated
as true.\end{itemize} \end{proof}

As an example let $\phi: c_1 \wedge c_2 \wedge c_3$ and $c_1= x_1 \vee y_1\vee z_1$,
$c_2= x_2 \vee u_2 \vee v_2 $ and $c_3= y_3 \vee a_3 \vee b_3$, where
$x=b=TRUE$ and $y=z=u=a=v=FALSE$. 
The graph
shown in Figure~\ref{fig:ill} is the square root of $G(\phi)$.

	\begin{figure}
	\begin{center}
	\includegraphics[scale=.41]{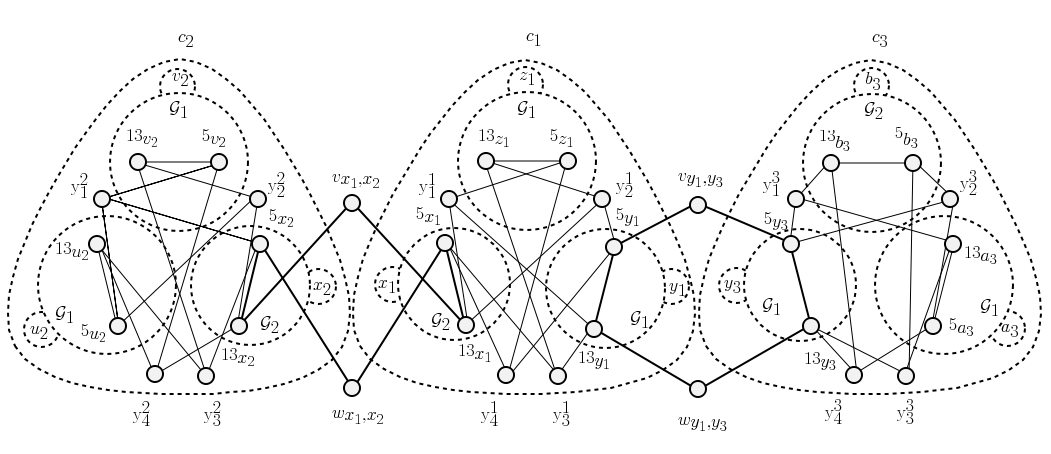}
	\caption{An example of $H$.}
	\label{fig:ill}
	\end{center}	
	\end{figure}
	
\begin{theorem}
\textsc{Square of Graphs With Girth Five} is NP-complete.
\end{theorem}

\begin{theorem}[The Complete Dichotomy Theorem]
\label{Theorem:mine}
\textsc{Square of Graphs With Girth $g$} is NP-complete if and only if $g \leq 5$.
\end{theorem}

\section{Conclusions}
We have disproved the conjecture in \cite{Far2009} by showing that
\textsc{Square of Graphs With Girth Five} is NP-complete.
Together with results provided by Motwani and Sudan \cite{MotSud1994} and
Farzad et al. \cite{Far2009}, we presented Theorem~\ref{Theorem:mine} as a
complete dichotomy theorem for square root finding problem.
\par The problem of square root finding for graphs can be restated for higher roots.

\begin{center}
\textsc{$k^{th}$ Power of a Graph With Girth $r$}
\begin{tabular}{lp{10cm}}
Instance: & A graph $G$.\\
Question: & Does there exists a graph $H$ with girth $r$ such that $G=H^k$.
\end{tabular}
\end{center}

\par The problem of root finding for higher root is an open problem
in terms of the $r^{th}$-root of the power graph.
Results provided by Adamaszek and Adamaszek \cite{Adam2010} is the closest
result to a complete girth-parametrized complexity dichotomy.
They proved that the recognition problem of \textsc{$k^{th}$Power of a Graph With Girth $r$}
is NP-complete when $r=k$ while there is a polynomial time
algorithm to find all $k^{th}$-roots of girth $2k+3$ for a given graph.
\par The problem of finding a complete girth-parametrized complexity dichotomy for
\textsc{$k^{th}$Power of a Graph With Girth $r$} is open,
and we conjectured the following:
\begin{con}
\textsc{$k^{th}$Power of a Graph With Girth $r$} for $r=2k+1$ is NP-complete.
\end{con}

\newpage
\section*{Appendix}
\subsection*{A: Proof of Theorem~\ref{lemma:uniquepairofgeminis}}
Unique pair of square roots for $\mathfrak{G}$: In this appendix we show that
the graph $\mathfrak{G}={\mathcal{G}_{1}}^2(={\mathcal{G}_{2}}^2)$
has only two non-isomorphic square roots which are $\mathcal{G}_1$ and $\mathcal{G}_2$.
For the rest of this subsection we denote $\mathfrak{G}$ for the square of 
$\mathcal{G}_{1}$ (or $\mathcal{G}_{2}$).

\begin{proof}
\begin{lemma}
\label{lemma:uniqueneighbourhoodof1}
Let $\mathfrak{G}=H^2$ for $g(H)=5$, then $N_H (1)=\{7,11,12\}$.
\end{lemma}
\begin{proof}
We show this in the following four steps:
\begin{itemize}
	\item[I] $1 \overset{H}{\sim} 12$: Assume otherwise and let $1 \overset{H}{\nsim} 12$,
	now since $1 \overset{\mathfrak{G}}{\sim} 12$ then $N_{H}(1) \subseteq (N_{\mathfrak{G}}(12) \cap N_{\mathfrak{G}}(1))=\{11,7,2,3\}$.
	In other hand we have $11 \overset{\mathfrak{G}}{\sim} 7$, $2 \overset{\mathfrak{G}}{\sim} 3$ but non of $11$ or $7$ is not
	adjacent to any  of $2$ and $3$, therefore either $N_{H}(1) \subseteq \{2,3\}$ or
	$N_{H}(1) \subseteq \{7,11\}$. If $N_{H}(1) \subseteq \{2,3\}$ then we have a contradiction with
	$11 \in N_{\mathfrak{G}}(1)$, and if $N_{H}(1) \subseteq \{7,11\}$ we again have a contradiction with
	$2 \in N_{\mathfrak{G}}(1)$, this implies $1 \overset{H}{\sim} 12$.
	\item[II] $1 \overset{H}{\sim} 11$: Assume otherwise and let $1 \overset{H}{\nsim} 11$, 
	now since $1 \overset{\mathfrak{G}}{\sim} 11$ then $N_{H}(1) \subseteq (N_{\mathfrak{G}}(11) \cap N_{\mathfrak{G}}(1))=\{7,5,12,13,14\}$.
	But according to part $I$, we know that $1 \overset{H}{\sim} 12$, therefore $N_{H}(1) \subseteq \{7,12\}$,
	and this is a contradiction because non of $12$ and $7$ are not adjacent to $14$ in $\mathfrak{G}$, so it
	implies $1 \overset{H}{\sim} 11$.
	\item[III] $1 \overset{H}{\sim} 7$: Assume otherwise and let $1 \overset{H}{\nsim} 7$,
	now since $1 \overset{\mathfrak{G}}{\sim} 7$ then $N_{H}(1) \subseteq (N_{\mathfrak{G}}(7) \cap N_{\mathfrak{G}}(1))=\{5,11,12,13\}$.
	Again according to part $I$ and $II$,$1 \overset{H}{\sim} 11, 12$, therefore $N_{H}(1) = \{11,12\}$.
	Here we have two possibilities, either $7 \overset{H}{\sim} 11$ or $7 \overset{H}{\sim} 12$.
	If $7 \overset{H}{\sim} 11$, since $14 \overset{\mathfrak{G}}{\nsim} 12$ then $14 \overset{H}{\sim} 11$ and this is
	a contradiction since $7 \overset{\mathfrak{G}}{\nsim} 14$. If $7 \overset{H}{\sim} 12$, since $2 \overset{\mathfrak{G}}{\nsim} 11$
	then $2 \overset{H}{\sim} 12$ and this is a contradiction since $7 \overset{\mathfrak{G}}{\nsim} 2$.
	So it implies $1 \overset{H}{\sim} 7$.
	\item[IV] $N_{H}(1) = \{7,11,12\} $: Since $\{7,11,12\}$ is a maximal clique in $\mathfrak{G}$,
	and $\{7,11,12\} \subseteq N_{H}(1)$ therefore $N_{H}(1) = \{7,11,12\} $.
\end{itemize}
\end{proof}

For more convenient
we use the following notation. For $v \in V_\mathfrak{G}$ let $x \in N_{\mathfrak{G}}(v)-N_{H}(v)$, we define $L_v (x)$ as follows:
\begin{center}
$L_v (x)=\{u\in N_{H}(v) \mid ux \in E_{\mathfrak{G}}\}=N_{\mathfrak{\mathfrak{G}}}(x) \cap N_{H}(v)$
\end{center}.
Since girth of $H$ is $5$ for all vertices $v$ and $x$ where $x \in N_{\mathfrak{G}}(v)-N_{H}(v)$,
there is a unique $u \in L_v(x)$ such that $u \overset{H}{\sim} x$.\\
According to Lemma~\ref{lemma:uniqueneighbourhoodof1} we have:
\begin{itemize}
\item[-]$N_{H}(12)=\{2,3,1\}$, since $L_{1}(2)=L_{1}(3)=12$.
\end{itemize}
Also $14 \overset{H}{\sim} 11$ (because $L_{1}(14)=11$), but 
$L_{1}(5)=L_{1}(13)=\{7,11\}$, hence we have two possibilities:
\begin{itemize}
	\item[I] Case 1: $5 \overset{H}{\sim} 11$ and $13 \overset{H}{\sim} 7$:
	\begin{itemize}
		\item[-] $N_{H}(7)=\{1,11\}$: trivial.
		\item[-] $N_{H}(11)=\{1,5,14\}$: trivial.
		\item[-] $N_{H}(13)=N_{G}(7)-N_{H}(7)-\{11,12\}=\{5,6,16\}$: since $13$ and $1$
		are the only neighbours of $7$.\\
		We now consider the set $N_\mathfrak{G}(13)-N_H(13)= \{1,3,8,10,11,14,15\}$,
		we have $L_{13}(3)=\{16\},$ $ L_{13}(8)=\{6,16\},$ $L_{13}(10)=\{16\},$ $L_{13}(14)=\{5,6,16\}
		,$ $ L_{13}(15)=\{5\}$,
		therefore:
		\item[-] $N_{H}(5)=\{13,15,11\}$, since $14 \overset{\mathfrak{G}}{\sim} 11$ (otherwise we have a cycle of length four).
		\item[-] $N_{H}(16)=\{13,3,8,10\}$, since $14 \overset{\mathfrak{G}}{\nsim} 3$.
		\item[-] $N_{H}(6)=\{13,14\}$, trivial.\\
		We now consider the set $N_G(12)-N_H(12)= \{8,9,15,16\}$,
		we have $L_{12}(8)=L_{12}(9)=L_{12}(15)=L_{12}(16)=\{2,3\}$, however we know
		that $16 \overset{H}{\sim} 3,8$:
		\item[-] $N_{H}(2)=\{12,15,8\}$, since $16 \overset{H}{\sim} 8$ (otherwise we have a cycle of length four),
		and $9 \overset{H}{\nsim} 8$.
		\item[-] $N_{H}(3)=\{12,9,16\}$, trivial.\\
		\item[-] $N_{H}(15)=\{2,4,5,9\}$, considering $N_{\mathfrak{G}}(5)-N_{H}(5)$ and refining the known neighbours.
		\item[-] $N_{H}(14)= \{4,6,8,11\}$, similar argument to vertex $15$, considering the vertex $11$.
		\item[-] $N_{H}(4)= \{10,14,15\}$, similar argument to vertex $15$, considering the vertex $15$.
		\item[-] $N_{H}(8)= \{2,14,16\}$, trivial.
		\item[-] $N_{H}(9)= \{3,15\}$, trivial.
		\item[-] $N_{H}(10)= \{4,16\}$, trivial.\\
		It can be seen that the above graph is $\mathcal{G}_1$.
	\end{itemize}
	\item[II] Case 2: $5 \overset{H}{\sim} 7$ and $13 \overset{H}{\sim} 11$:
	\begin{itemize}
		\item[-] $N_{H}(7)=\{1,5\}$: trivial.
		\item[-] $N_{H}(11)=\{1,13,14\}$: trivial.
		\item[-] $N_{H}(5)=N_{G}(7)-N_{H}(7)-\{1,7\}=\{6,7,13,16\}$: since $5$ and $1$
		are the only neighbours of $7$.\\
		We now consider the set $N_G(5)-N_H(5)= \{1,2,4,9,11,14,15\}$,
		we have $L_{5}(2)=\{16\},$ $ L_{5}(4)=\{6,16\},$ $L_{5}(9)=\{16\},$ $L_{5}(14)=\{6,13,16\}
		,$ $ L_{5}(15)=\{13\}$,
		therefore:
		\item[-] $N_{H}(13)=\{5,11,15\}$, since $14 \overset{\mathfrak{G}}{\sim} 11$ (otherwise we have a cycle of length three).
		\item[-] $N_{H}(16)=\{2,4,5,9\}$, since $14 \overset{\mathfrak{G}}{\nsim} 9$, and also $4 \overset{\mathfrak{G}}{\nsim} 13$ but
		$4 \overset{\mathfrak{G}}{\sim} 11$, therefore $4 \overset{H}{\nsim} 6$.
		\item[-] $N_{H}(6)=\{5,14\}$, trivial.\\
		We now consider the set $N_G(12)-N_H(12)= \{8,9,15,16\}$,
		we have $L_{12}(8)=L_{12}(9)=L_{12}(15)=L_{12}(16)=\{2,3\}$, however we know
		that $16 \overset{H}{\sim} 3,8$:
		\item[-] $N_{H}(2)=\{12,15,8\}$, since $16 \overset{H}{\sim} 8$ (otherwise we have a cycle of length four),
		and $9 \overset{H}{\nsim} 8$.
		\item[-] $N_{H}(3)=\{12,9,16\}$, trivial.\\
		\item[-] $N_{H}(15)=\{3,8,10,13\}$, considering $N_{\mathfrak{G}}(5)-N_{H}(5)$ and refining the known neighbours.
		\item[-] $N_{H}(14)= \{4,6,8,11\}$, similar argument to vertex $15$, considering the vertex $11$.
		\item[-] $N_{H}(4)= \{10,14,16\}$, similar argument to vertex $15$, considering the vertex $15$.
		\item[-] $N_{H}(8)= \{2,14,15\}$, trivial.
		\item[-] $N_{H}(9)= \{3,16\}$, trivial.
		\item[-] $N_{H}(10)= \{4,15\}$, trivial.\\
		It can be seen that the above graph is $\mathcal{G}_2$.
	\end{itemize}
\end{itemize}
So $H$ is either isomorphic to $\mathcal{G}_1$ or $\mathcal{G}_2$.
\end{proof}
\end{document}